\RequirePackage{fix-cm}
\documentclass[smallextended]{svjour3}       
\smartqed  
\usepackage{graphicx}
\usepackage[round]{natbib}
\usepackage{mathptmx}      
\usepackage{amsmath}
\usepackage{algorithm}
\usepackage{algorithmic}
\usepackage{amssymb}

\renewcommand{\cite}{\citep}

\newtheorem{thm}{Theorem}
\newtheorem{prop}{Proposition}
\newtheorem{defn}{Definition}
 \journalname{Computational Economics}

\begin{document}

\title{Accelerating Implicit Finite Difference Schemes Using a Hardware Optimized Tridiagonal Solver for FPGAs
}

\titlerunning{Hardware Optimized Tridiagonal Solver for FPGAs}        

\author{Samuel Palmer 
}


\institute{S. Palmer \at
              Department of Computer Science, University College London, Gower Street, London, WC1E 6BT \\
              \email{ucabsdp@ucl.ac.uk}           
}

\date{Received: date / Accepted: date}

\maketitle

\thispagestyle{empty}
\pagestyle{empty}

\begin{abstract} 

We present a design and implementation of the Thomas algorithm optimized for hardware acceleration on an FPGA, the Thomas Core. The hardware-based algorithm combined with the custom data flow and low level parallelism available in an FPGA reduces the overall complexity from 8N down to 5N serial arithmetic operations, and almost halves the overall latency by parallelizing the two costly divisions. Combining this with a data streaming interface, we reduce memory overheads to 2 N-length vectors per N-tridiagonal system to be solved. The Thomas Core allows for multiple independent tridiagonal systems to be continuously solved in parallel, providing an efficient and scalable accelerator for many numerical computations. Finally we present applications for derivatives pricing problems using implicit finite difference schemes on an FPGA accelerated system and we investigate the use and limitations of fixed-point arithmetic in our algorithm.


\acknowledgement{I would like to thank David Thomas for his support with this work.}

\end{abstract}

\section{Introduction}
\subsection{FPGAs}
Field programmable gate arrays (FPGAs) provide an integrated circuit that can be reconfigured on the fly or `in the field' in the form of a chip.  FPGAs provides a flexible and cost effective way to develop and implement custom hardware designs. The core component that allows an FPGA to be reconfigurable is a look-up table (LUT). A LUT produces one or more outputs as a function of the digital inputs. These functions are determined when the FPGA is configured and provide the desired logic controlled via the programmable cells. The other key component that helps to increase the performance of FPGAs is on-chip block memory (BRAM) which can provide fast local memory caches. Some FPGA chips offer additional features such as high speed digital signal processors (DSP) and multipliers. The FPGA chip is then usually embedded on a circuit board and connected to additional peripherals such as DDR memory, USB ports, ethernet, PCI express and VGA ports to provide the complete heterogeneous computing system. 

\subsection{Finite Difference Schemes and Tridiagonal Systems}

Finite difference (FD) schemes are an important tool for solving parabolic partial differential equations (PDEs) numerically. In financial engineering FD methods are commonly employed to solve PDEs which model derivatives, such as the famous Black-Scholes equation (BSE) \cite{Tavella}.

Finite difference schemes begin by discretising the problem domain into a mesh/grid over the time interval $[0, 1]$, and in basic cases, the asset price interval $[0,S_{max}]$. There are several variations and enhancements, however we now describe the basic implicit scheme for one dimension. The domain is discretised into $N$ asset price steps and $M$ time steps, given by:\\
\begin{eqnarray}
\Delta S = \frac{S_{max}}{N}\\
\Delta t = \frac{1}{M}
\end{eqnarray}

\noindent The spatial derivative terms are approximated using central differencing, and the time derivative using backwards differencing. These discretizations are then substituted into the PDE to produce the discrete difference equation. For example the BSE gives:

\begin{eqnarray} \label{eq:implicit}
V_n^{m} = a_nV_{n-1}^{m-1} + b_nV_{n}^{m-1}+c_nV_{n+1}^{m-1}
\end{eqnarray}

\noindent with problem dependant stencil coefficient values $a_n$, $b_n$, $c_n$. The resulting system of equations can be written in matrix form as

\begin{eqnarray}
{\bf{A}}{\bf{V}}^{t-1} = {\bf{V}}^{t}
\end{eqnarray}

\noindent This matrix inversion problem involves solving for the price vector ${\bf{V}^{t-1}}$ at the current time-step, where the vector $V^t$ is known from the previous implicit step. The coefficient matrix ${\bf{A}}$ takes on the banded tridiagonal form shown below:
\begin{eqnarray}\label{eq:banded}
{\bf{A}} = 
\begin{bmatrix}
b_0 & c_0 & 0 & 0 & 0 &...\\
a_1 & b_1 & c_1 & 0 & 0 &...\\
0 & a_2 & b_2 & c_2 & 0 &...\\
0 & 0 & a_3 & b_3 & c_3 & 0 &...\\
.\\
.\\
.\\
0 &  &  & & ... & 0 & a_N & b_N 
\end{bmatrix}
\end{eqnarray}

\noindent or more generally written as the matrix inversion problem $\bf{A}x = y$. When pricing multidimensional derivatives, such as basket options or under stochastic volatility, another class of finite difference schemes known as alternating-direction-implicit schemes \cite{Peaceman1955} may be used. These schemes solve the PDE in an implicit manner within multiple dimensions. These methods can be computational challenging as they require solving multiple tridiagonal systems at each time step, thus significant research efforts have gone into creating fast parallel solvers on devices such as GPUs \cite{Dang2010} \cite{Egloff2011}.

\subsection{Thomas Algorithm}

\begin{algorithm}[H]
\caption{Thomas Algorithm (a,b,c,y) Pseudo Code}
\center
\begin{algorithmic}
\STATE d[0] = b[0]
\STATE z[0] = y[0]
\FOR{ i = 1 to $N$} 
\STATE prev = i - 1
\STATE $l_i$ = a[i]/d[prev]
\STATE d[i] = b[i]-$l_i$*c[prev]
\STATE z[i] = y[i] - $l_i$*z[prev]
\ENDFOR
\STATE z[N] = z[N]/d[N]
\FOR{ i = N-1 to 0} 
\STATE x[i] = (z[i]-c[i]*x[i+1])/d[i]
\ENDFOR
\STATE return x[i]
\end{algorithmic}
\end{algorithm}

The Thomas algorithm \cite{thomas} is the simplest method used to solve a tridiagonal system of equations and is commonly employed on serial devices such as a CPU. The Thomas algorithm is a specialised case of Gaussian elimination and can be derived from the LU decomposition of the matrix $A$. This reduces the system to the solution of two bi-diagonal systems which can then be solved via Gaussian elimination. The first system is solved via forward substitution and the second system is solved via backward substitution. These two stages will be referred to as the forward and backward iterations. The Thomas algorithm is given in Algorithm 1, it has a complexity of $O(N)$ and requires a total of 8N arithmetic operations to solve an N-tridiagonal system.

In parallel computing the Thomas algorithm is usually less favoured compared to algorithms such as recursive-doubling \cite{Stone1973}, cyclic-reduction \cite{Hockney1965} and parallel cyclic-reduction \cite{Hockney1981}, while these algorithms have a larger number of arithmetic operations, some of the operations be can parallelised on devices such as GPUs \cite{Zhang2010} resulting in an overall lower algorithmic complexity. With recent increased interest in FPGA acceleration, attempts have been made to port tridiagonal solvers onto FPGAs \cite{Oliveira2008, Warne2012, warne2014, Chatziparaskevas2012}. In this application the simplicity of the Thomas algorithm makes it well suited to the task when compared to cyclic-reduction which may be too complex for efficient data flow FPGA implementation.

\section{Algorithmic Optimization and Low Level Parallelizm}
Figures \ref{fig:dpf} and \ref{fig:dpb} depict the data dependency of the Thomas algorithm. We observed that in the forward iteration there are two separate branches of computation to calculate$d_n$ and $z_n$ which suggests a first level of parallelism. A similar approach has been taken by both Oliveira et al \cite{Oliveira2008} and Warne et al \cite{Warne2012}. This optimisation reduces the effective serial arithmetic operations down from 8N to 6N. 

\begin{figure}
\begin{center}
\includegraphics[scale=0.10]{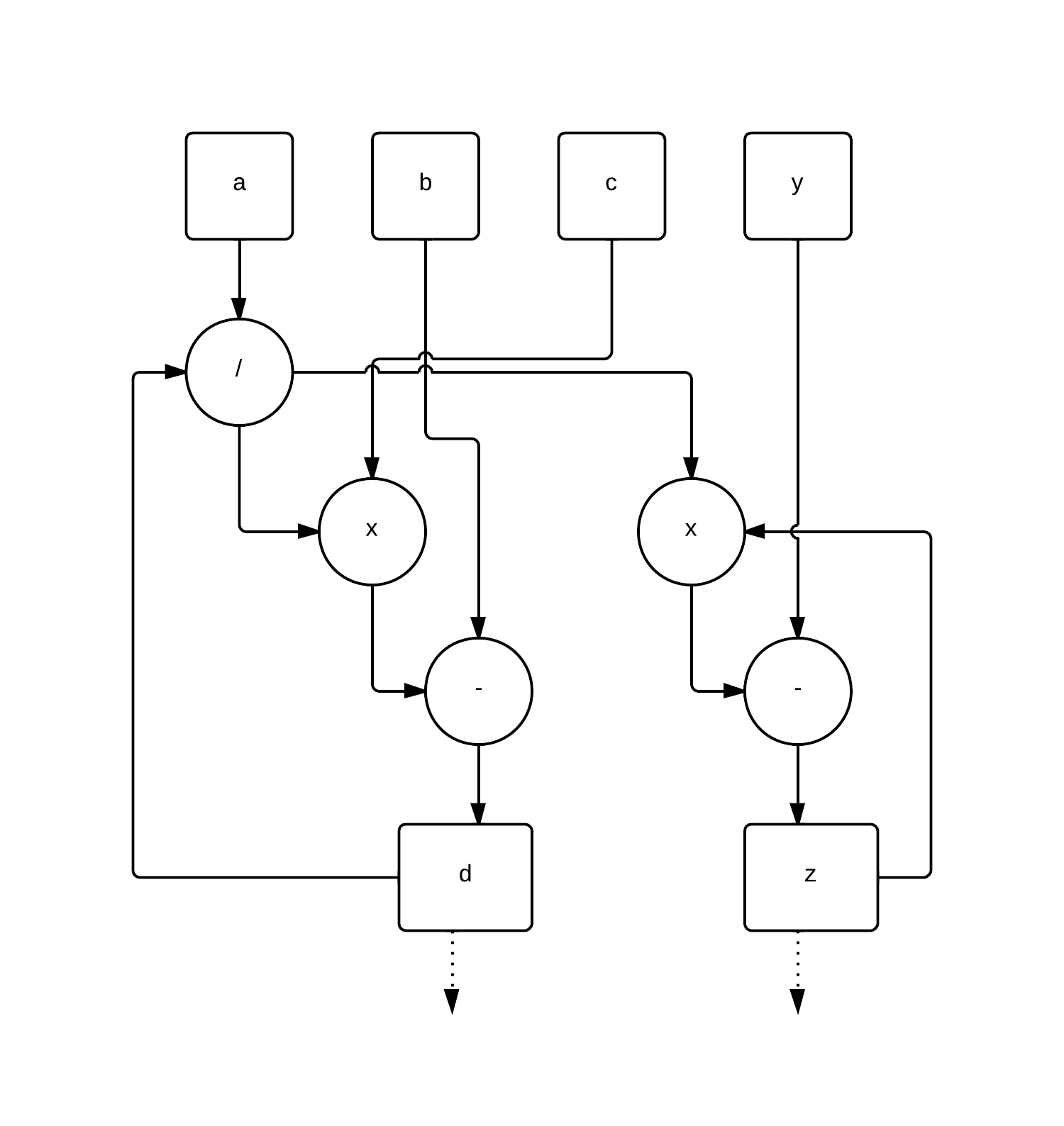}
\caption{Data dependency graph for the forward iteration of the Thomas algorithm}
\label{fig:dpf}
\end{center}
\end{figure}

\begin{figure}
\begin{center}
\includegraphics[scale=0.10]{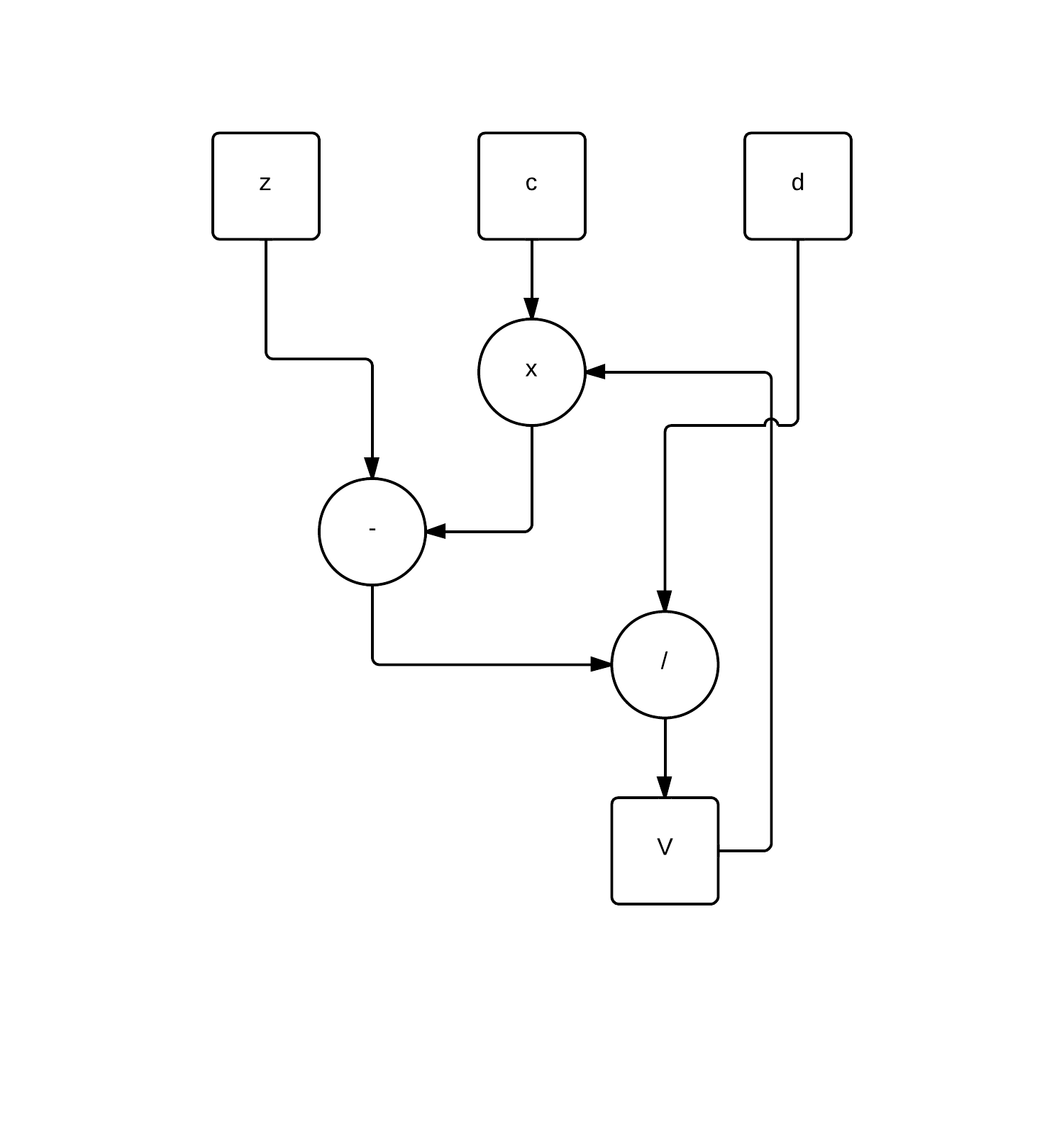}
\caption{Data dependency graph for the backwards iteration of the Thomas algorithm}
\label{fig:dpb}
\end{center}
\end{figure}

The problem with this simple optimisation is that while there is a reduction in serial operations these are all subtractions or multiplications, which are computational cheap when compared to divisions. Consequently a competitive speed-up over faster clocking devices such as CPUs may not be obtained \cite{Warne2012}. We thus introduce a simple algorithmic rearrangement that can allow for the two divisions from the backwards and forward iterations to be effectively parallelised. Equation \ref{eq:factorise} shows the factorisation of the backwards iteration calculation where we now treat the divisions of $z_n$ and $c_n$ by $d_n$ individually
\begin{equation}\label{eq:factorise}
\frac{z_n - c_nV_{n+1}}{d_n} = \frac{1}{d_n}z_n - (\frac{1}{d_n}c_n)V_{n+1} .
\end{equation}

\begin{figure}
\begin{center}
\includegraphics[scale=0.10]{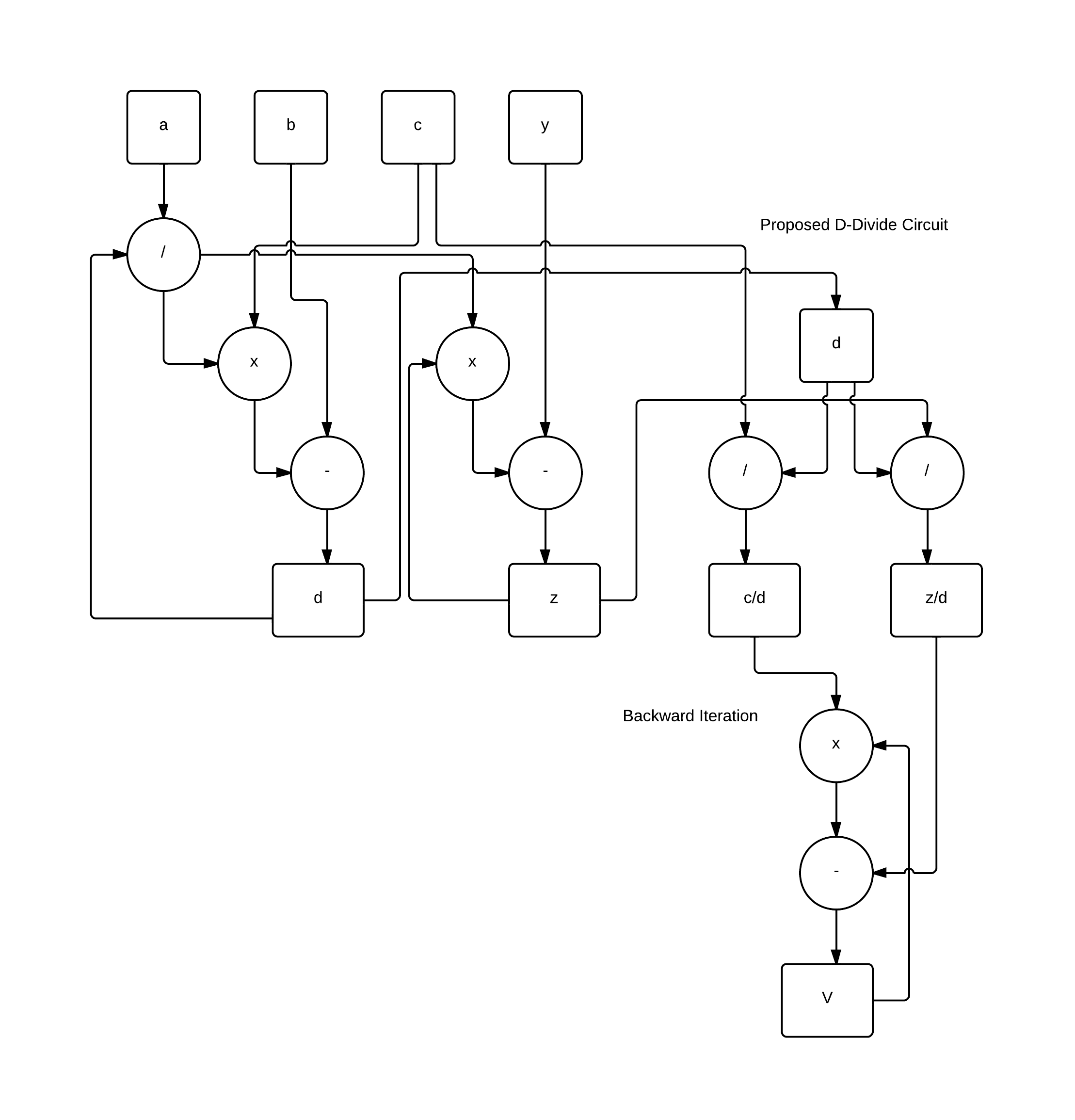}
\caption{Data dependency graph for the proposed Thomas algorithm structure optimised for FPGA implementation.}
\label{fig:proposedFlow}
\end{center}
\end{figure}

In a serial implementation this would add an extra division to the total number of arithmetic operations, which is not usually desirable, but as shown in Figure \ref{fig:proposedFlow} the data dependence is in fact reduced in this framework. Indeed we can now treat these two divisions in parallel with each other as well as in parallel with the forward iteration calculations. This reduces the serial arithmetic operations down from the original 8N to 5N.

For FPGA implementations this rearranged algorithm has two advantages:
\begin{enumerate}
\item  Total latency of the algorithm is almost halved by parallelising the two lengthy divisions.
\item Memory requirements are reduced from the need to save three intermediate data vectors ($c$, $z$ and $d$) to two ($c/d$ and $z/d$).
\end{enumerate}

\subsection{Pipelining}

Further to the low level algorithmic optimisations higher level parallelism can be achieved in two ways: pipelining the data through the computation of the forward and backward iterations; and pipelining the sets of data between the forward and backward iterations, which has commonly been implemented for multiple CPU versions to parallelise the Thomas algorithm \cite{Duff1999}.

Firstly the computational units themselves can be deeply pipelined, an approach used by Olivera et al \cite{Oliveira2008}, which allows for multiple independent tridiagonal systems to be computed in the same iteration cycle. For example, if the forward iteration computational unit has $P_F$ pipeline stages then throughout one iteration it is possible to fill each stage of the pipeline with a computation allowing for $P_F$ independent tridiagonal systems to be computed. 

The second type of pipelining is that given a set $T$ of pipelined tridiagonal systems for either iteration, as discussed above, we can simultaneously compute the forward and backwards iterations of the two different sets (that the first set has already been through the forward iterations) independently in one Thomas solver. In \cite{warne2014} the authors investigate using OpenCL and Xilinix HLS to build the Thomas solvers, but do not obtain this level of parallelism due to the complex scheduling involved for the pipelines. In order to achieve the desired low level control we implemented this design in VHDL.

\subsection{Hardware Architecture}

The input to the solver core consists of 5 data items $a, b, c, y$ and $id$ where $id$ is the local identifier of the system to be solved. This acts as a thread identifier and is important for addressing the correct memory stacks in the solver. The hardware architecture consists of four main components: the forward iteration core, the d-divider, the backwards iteration core and the stack array. The forward and backward cores contain the pipelined arithmetic for the stages of the algorithm, and the d-divider consists of two dividers to carry out the $c/d$ and $z/d$ computations. The stack array is used for storing the intermediate variables $c/d$ and $z/d$. A stack can be used due to the nature of the problem as the backwards iterations first require the last values calculated by the d-divider, which saves unnecessary complications with memory addressing. 

Connecting the forward iterations to the backwards iterations is a queue. This queue allows the problem index to be passed onto the backwards core for computation once the forward iterations have finished, this system allows for efficient independent operation of the forwards and backward iterations. The backwards core checks for space in the pipeline, and reads in the problem to begin computing if there is space, otherwise it remains queued. 

In addition to the main Thomas algorithm core, the core is placed in a wrapper allowing for easy usability. The wrapper consists of first-in-first-out (fifo) queues for the input data and output results, allowing for variable write and read times to and from the core, as well as the option for floating-point to fixed-point converters for input data and vice-versa for results. 

When changing the arithmetic only the arithmetic cores are changed, and the architecture remains constant. The only variability with the arithmetic cores is the pipelining due to the differing latencies, but this is managed via adjustable parameters within the solver VHDL. 

\section{Design Analysis}
Here we theoretically analyse the performance of the solver for solving multiple independent tridiagonal systems $\bf{T}$ $ = \{T^{N_1}_1, T^{N_2}_2, \ldots , T^{N_m}_M\}$, where M is the total number of independent tridiagonal systems to be solved and $N_m$ is the size of the $m$th system. The notation is used in this work is as follows:\\

\begin{itemize}
\item $T_m^{N_m}$ is the size of the $m$th tridiagonal system to be solved of size $N_m$,\\
\item $\bf{T^N}$ is a special case where for all $T_m \in$ $\bf{T}$, $N_m = N$,\\
\item $C^D_{\{+,-,/,\times\}}$ is the number of clock cycles taken for that arithmetic operation using data format $D$,\\
\item $C_{F,B,A}$ is the number of clock cycles taken for a single forwards ($F$) and backwards ($B$) iteration and administration costs ($A$).\\
\item $f$ is the clock frequency of the FPGA system.\\
\end{itemize}

The number of cycles taken for the iteration stages are:
\begin{eqnarray}
C_F = C^D_/ + C^D_\times + C^D_- \\
C_B = C^D_\times + C^D_+ + C^D_/ ,
\end{eqnarray}
\noindent where $C_A$ is a constant determined by the programming of the algorithm.\\

\noindent To fully harness the power of the pipelined design it is desired that maximal throughput should be achieved by scheduling groups of independent computations. \\

\begin{defn}
The number of computational blocks, $B$, is defined as the number of subsets of independent tridiagonal systems to be solved. The set of of blocks given by $\bf{B}$ = $\{{\bf{b}}^{m_1}_1,{\bf{ b}}_2^{m_2}, \ldots ,{\bf{b}}_B^{m_B}\}$, where ${\bf{b}}^{m_i}_i \subset$ $\bf{T}$ has size $m_i$ such that:
\begin{eqnarray}
\notag
\cup_{i=1}^B {\bf{b}}^{m_i}_i = {\bf{T}}  \quad \text{ }\cap_{i=1}^B {\bf{b}}^{m_i}_i = \emptyset .
\end{eqnarray}
\end{defn}

 \noindent Thus for a given $M$ the time to compute $\bf{T^N}$ is given by: 
\begin{eqnarray}\label{eq:speed}
t_{\bf{T^N}} =  \frac{ NB(C_F + C_A)   + BC_/ + NC_B + 2\sum_{b=1}^B (m_b-1) }{f} .
\end{eqnarray} 
\normalsize

The partitioning of $\bf{T}$ into the set of blocks $\bf{B}$ can be affected by the data transfer rate $r_d$ between the solver and the host system. The rate of computation, $r_c$, of the Thomas solver is given by:
\begin{eqnarray}
r_c = \frac{5D}{f} ,
\end{eqnarray}
\noindent where $D$ is the number of bits used to represent a number in the given format. This value is the rate at which data can be processed by the Thomas solver, which requires 5 inputs, $a, b, c, y$ and $id$ and can process a row every clock cycle.

The optimal number of blocks $B$ can be obtained if the rate of transfer is quicker than the rate of computation, i.e. the solver can receive all the 5 values in one clock cycle or less: \\
\begin{eqnarray}
B_{\text{opt}} = \text{floor} \left( \frac{M}{C_F} \right) \text{ , } r_d \geq r_c
\end{eqnarray}

\noindent It may be the case that the data transfer rate is slower than the rate of computation and hence the solver has to be stalled waiting for the data. It is therefore desirable to compute the maximum number of tridiagonal systems in a block $\bf{b}$ in the pipeline without stalling for data. The number of blocks $B$ is given by: 

\begin{align}
m_{\text{opt}} &= \text{ceil}(\frac{r_c}{r_d}) \text{ , } r_d < r_c \\
B &= \text{floor} \left( \frac{M}{m_{opt}} \right)
\end{align}

\noindent Maximum throughput for the solver can be obtained if the set of tridiagonal systems to be solved completely fills the pipeline of the solver:
\begin{align}\label{eq:maxtp}
M \text{ mod } C_F &= 0 \\
B &> 1 .
\end{align}

\section{Fixed-Point Arithmetic Analysis}

\subsection{Numerical Bounds}
To maximise performance it may be required that custom data formats are used in the FPGA design. Fixed-Point arithmetic often provides faster and smaller FPGA designs, for example \cite{Jin2012, Tian2010}, but at the cost of the loss of some precision in the results and a higher risk of arithmetic overflow. Therefore it is important to know the range of values the solver is expected to use in the algorithm to allow for the custom data formats to be optimised for the problem. The preceding results here require that  $b(n)$ is a positive, monotonically increasing function of the row index, $n$, i.e. $b_n < b_{n+1}$ and $|a_n|< 1 $ and $|c_n|<1$ $\forall i \leq N$. These theorems will be useful later for range bounding the implicit pricing problem. In these following results the $\ell_{\infty}$-norm of the set of coefficients $a$, $b$ or $c$, denoted by $\|x\|_{\infty}$, is used, the value of this norm is the largest absolute value in a set.

\begin{thm}\label{thm:dbound}
Let $A$ be diagonally dominant by row or columns, and let $A$ have LU factorisation $A=LU$. Then $\| d \|_\infty \leq 3\| b \|_\infty$
\end{thm}
\begin{proof}
We use the result that: 
	\begin{equation}
	|l_n c_{n-1}| + |d_n| \leq 3|b_n| 
	\end{equation}
\noindent see \cite[pg.175]{Higham2002}. Simple rearrangement and the observation that the maximum will occur at the maximum absolute value gives the result of Theorem \ref{thm:dbound}
\end{proof}

\begin{thm}\label{thm:dlower}
Let $A$ be diagonal dominant by row, and let $A$ have LU factorisation $A=LU$ then 
$|d_n|  > |b_0 - \frac{\|a\| \infty}{|b_0|}\|c\|_{\infty}| \quad \text{for every } n$, providing that $b(n)$ is a positive monotonically increasing function of the row index, $i$ and $\Delta b \leq \|c\|_{\infty}$.
\end{thm}
\begin{proof}
	\begin{align}
	d_0 &= b_0\\
	d_1 &= b_1 - \frac{a_1}{b_0}c_0\\
	b_1 &- \frac{\|a\|_{\infty}}{|b_0|}\|c\|_{\infty} \leq d_1
	\end{align}
\noindent Under the assumption that $b(n)$ is a positive monotonically increasing function we have

\begin{eqnarray}
b_0 - \frac{\|a\|_{\infty}}{|b_0|}\|c\|_{\infty} \leq b_1 - \frac{\|a\|_{\infty}}{|b_0|}\|c\|_{\infty} \leq d_1 .
\end{eqnarray}

\noindent Finally for this to hold over all cases it must hold that $\|l\|_{\infty} \leq \frac{\| a\|_{\infty}}{|b_0|}$, which implies that:

\begin{eqnarray}
b_0 \leq b_1 - \frac{\|a\|_{\infty}}{|b_0|}\|c\|_{\infty} ,
\end{eqnarray}

\noindent thus for this for hold $\Delta b \leq \|c\|_{\infty}$, given that $\|a\|_{\infty} < |b_0|$.\\
\end{proof}

Guarantees may be available for more general functions $b(n)$, however Theorem \ref{thm:dlower} currently suffices for our purposes. A possible approach may investigate under what conditions a certain sequence $b$ minimises $d_0$.

\begin{thm}\label{thm:lupper}
Let $A$ is diagonal dominant by row, and let $A$ have LU factorisation $A=LU$ then 
$\| l \|_{\infty} < \frac{\|a\|_{\infty}}{ |b_0 - \frac{\|a\|_{\infty}}{|b_0|}\|c\|_{\infty}| } < \frac{\|a\|_{\infty}}{|b_0|}$, provided that $b(n)$ is a positive monotonically increasing function of the row index, $n$, and $\Delta b \leq \|c\|_{\infty}$.
\end{thm}
\begin{proof}
Using Theorem \ref{thm:dlower}, the maximum value of $l$ must be achieved when the largest value of $a$ is divided by the smallest value of $d$.
\end{proof}

\noindent In fact, although Theorem \ref{thm:dbound} provides an upper bound for the value of $d$, using the previous theorems a tighter more accurate bound can now be defined.

\begin{thm}\label{thm:duppertight}
Let $A$ is diagonal dominant by row, and let $A$ have LU factorisation $A=LU$ then 
$d \leq \|b\|_{\infty} + \| l \|_{\infty} \| c \|_{\infty}$, provided that $b(n)$ is a monotonically increasing function of the row index $n$. \\
\end{thm}

\subsubsection{Bounding the Thomas Algorithm}

The first section describes various bounds for the LU decomposition of a matrix. This forms the basis of the well known Thomas Algorithm used for solving tridiagonal inversion problems of the form $Tx=y$, where $T$ is a tridiagonal matrix. The first stage of the algorithm is to apply LU decomposition to the matrix and then solve to auxiliary equations using forward and backward substitution. 

\begin{thm}\label{thm:fwrdBound}
Let a tridiagonal matrix, $T$, which is diagonally dominant by row, have LU factorisation $T=LU$ with $\| l \|_{\infty} < 1$. Then solving the first auxiliary equation of the inversion problem $Lz = y \text{ with } z = Ux $, it holds that 
$\| r \|_{\infty} < \|y\|_{\infty}( \frac{1}{1-\| l \|_{\infty}})$
\end{thm}
\begin{proof}
First the term for $z_N$ is expanded and using theorem \ref{thm:lupper} it is possible to replace the individual $l_i$ terms with the upper bound $\| l \|_{\infty}$
\begin{equation}
\| r \|_{\infty} \leq |y_N| + \sum_{k=1}^{N-1} \| l \|_{\infty}^k|y_{k}| .
\end{equation}

It is then possible to compact the telescopic sum into a geometric sequence, which has a maximum value when $i = N$, that is using all of the terms. Given that the index of the largest $y$ value may not be known a larger bound can be formed by including this in the geometric sum as the final term.\\

\noindent In fact we can further loosen the bound by assuming that all values are the maximum, so 
\begin{equation}
\| z \|_{\infty} \leq |y_N| + \sum_{k=1}^{N-1} \| l \|_{\infty}^k|y_{k}| \leq \|y\|_{\infty}(1+ \sum_{k=1}^{N-1} \| l \|_{\infty}^k) .
\end{equation}
\noindent Using the formulas for the sum of a geometric sequence the max bound becomes
\begin{equation}
\| z \|_{\infty} \leq \|y\|_{\infty} \left( 1+ \frac{\| l \|_{\infty}(1-\| l \|_{\infty}^{N-1})}{1-\| l \|_{\infty}} \right) ,
\end{equation}
\noindent and finally in the case that $\| l \|_{\infty} < 1$ a simpler form using the infinite geometric sum can be used:

\begin{equation}
\|y\|_{\infty} \left( 1+ \frac{\| l \|_{\infty} ( 1-\| l \|_{\infty}^{N-1} ) }{1-\| l \|_{\infty}} \right) < 
\|y\|_{\infty} \left( \frac{1}{1-\| l \|_{\infty}} \right)
\end{equation}

\end{proof}

The hardware optimised algorithm presented here requires the calculation of two additional intermediates $\frac{c}{d}$ and $\frac{z}{d}$.

\begin{thm}
Let $T$ be a tridiagonal matrix which is diagonally dominant by row and let $T$ have LU factorisation $T=LU$ with $\| l \|_{\infty} < 1$ and assume the conditions of Theorem \ref{thm:dlower} hold. Then the intermediate values, $|\frac{c}{d}| \leq \frac{\|c\|_{\infty}}{|b_0|}$ and $|\frac{z}{d}| \leq \frac{\|z\|_{\infty}}{|b_0|}$.
\end{thm}

We can now derive the bounds for the final values as follows.

\begin{thm}
Given a tridiagonal matrix $T$ is diagonally dominant by row and let $T$ have LU factorisation $T=LU$ with $\| l \|_{\infty} < 1$, and assume the conditions of Theorem \ref{thm:dlower} hold, with $b_n > 1$ and $c_n < 1$ $\text{for all } n\leq N$. Then $\|v\|_{\infty} < \frac{\|z\|_{\infty}}{|b_0| - 1}$ .
\end{thm}
\begin{proof}

\begin{eqnarray}
v_N = \frac{z}{d}\\
|v_N| < \frac{\|z\|_{\infty}}{|b_0|}
\end{eqnarray}

\noindent The recursion begins at $v_{N-1}$, hence it is possible to see that

\begin{eqnarray}
|v_{N-1}| < \frac{\|z\|_{\infty}}{|b_0|} + V_N \frac{\|c\|_{\infty}}{|b_0|} <  \frac{\|z\|_{\infty}}{|b_0|} + \frac{\|z\|_{\infty}}{|b_0|} \frac{\|c\|_{\infty}}{|b_0|} .
\end{eqnarray}

\noindent Expanding the recursion in this manner the result for a sequence is given by

\begin{eqnarray}
|v_n| < \sum_{i = 0}^{N-n}\frac{\|c\|_{\infty}^i\|z\|_{\infty}}{|b_0|^{i+1}} .
\end{eqnarray}

\noindent Now assuming $|b_0| > 1$ and $\|c\|_{\infty} <1$, the limit of this sequence is given by:

\begin{eqnarray}
|v_n| < \sum_{i = 0}^{N-n}\frac{\|c\|_{\infty}^i\|z\|_{\infty}}{|b_0|^{i+1}} < \frac{\|z\|_{\infty}}{|b_0|-1}
\end{eqnarray}
\end{proof}

Combing the previous theorems it is now possible to define a set of conditions that can ensure the absolute value of any variable in the algorithm does not exceed a certain bound. This will prove extremely useful when applying the fixed-point designs to given problems.
\begin{prop}\label{prop:tConditions}
For a given integer $Z > 0$ there exists a set of conditions such that all intermediate variables in the Thomas algorithm can be bounded by $Z$, given that  $b(n)$ is a positive monotonically increasing function of the row index, and $|a_n|< 1 $,  $|c_n|<1$ and $|b_n| > 1 $ $\forall n \leq N$. The following conditions are sufficient.  
\begin{enumerate}
\item  $\|l\|_{\infty} < 1$ this implies $\|a\|_{\infty} < |b_0| $
\item $\|y\|_{\infty} < Z\frac{|b_0| - \|a\|_{\infty}}{|b_0| + 1} $
\item $ \|c\|_{\infty} < |b_0| $
\item $\Delta b \leq \|c\|_{\infty} $
\end{enumerate}
\end{prop}
\noindent Note that the conditions in Proposition \ref{prop:tConditions} are sufficient but may not be necessary.


\section{Hardware Implementation}
\subsection{FPGA Resource Usage}
The Thomas Solver hardware was tested using the ZedBoard Xilinx Zynq7020 Evaluation Kit. The Zynq7020 is a system-on-chip which consists of two ARM-A9 processors connected to Xilinx Artix-7 FPGA fabric, allowing a high-speed interface between CPU and FPGA. Using the Zynq7020 the system of equations are formulated in floating-point on the ARM-A9 CPU, these are then transferred to the FPGA via AXI interfaces and solved using the FPGA Thomas solver. The fixed-point results are then converted back to floating-point and compared to the results for the same problem solved using floating-point arithmetic.

The results in Table \ref{tab:resources} have been obtained post-implementation from the Vivado Design Suite. The base design used has $N_{max} = 512$ with 10 threads ($M_{max} = 10$), and variable arithmetic. A floating point design and three fixed-point solvers with the data representation [integer bits, fractional bits] have been tested, 32bit[2,30], 24bit[2,22], 16bit[2,14]. For the arithmetic cores the provided Xilinx base IP cores have been used, and set to make maximum usage of DSPs, and the Radix-2 divider algorithm is used as part of the fixed-point divider. 

\begin{table}\label{tab:resources}
\center
\caption{FPGA resources used for each design and percentages of resources used on the Xilinx Zynq7020.}
\scriptsize
\begin{tabular}{lllll}
\hline
 &  \multicolumn{4}{c}{Arithmetic Format }\\
 & Floating  & Fixed[2,30] & Fixed[2,22] & Fixed[2,14] \\
 \hline

Flip-Flops & 25721 (24\%) & 15369 (14\%)  & 17224 (16\%) & 10711 (10\%) \\
LUT & 27204 (51\%) & 20722 (39\%) & 16998 (32\%) & 11894 (22\%) \\
Mem-LUT & 10547 (61\%) & 8683 (50\%) & 6174 (35\%) & 4294 (25\%) \\
BRAM & 35 (25\%) & 3 (2\%)  & 3 (2\%) & 3 (2\%) \\
DSP & 6 (3\%) & 15 (7\%)  & 9 (4\%) & 6 (3\%) \\
Buft & 1 (3\%) & 1 (3\%) & 1 (3\%) & 1 (3\%) \\
Clock & 100MHz & 200MHz & 200MHz & 200MHz\\
Power (W)  1.932 & 1.827  & 1.788 & 1.648 & 1.568 \\
\hline

\end{tabular}
\end{table}

\begin{table}\label{tab:latency}
\center
\caption{Clock cycle latency for each of the components of the Thomas solver core.}
\begin{tabular}{lllll}
\hline

 &  \multicolumn{4}{c}{Arithmetic Format }\\
 & Floating  & Fixed[2,30] & Fixed[2,22] & Fixed[2,14] \\
 \hline

Div (Radix-2) & 28 & 61 & 52 & 36 \\
Multiplier & 12 & 6 & 6 & 6 \\
Subtractor  & 4 & 2 & 2 & 2 \\
Thomas Forward  & 44 & 69 & 60 & 44 \\
Thomas Backward & 16 & 8 & 8 & 8 \\
Administration & 3 & 3 & 3 & 3 \\

\hline

\end{tabular}
\end{table}

Each of the solver designs has the same magnitude of latency, with the floating-point design providing the lowest total latency, although the fixed-point designs may be sped up by using higher radix divider algorithms. The disadvantage of the higher radix divider algorithms is that the maximum throughput is reduced due to the iterative nature of the algorithms, but this is useful if it is not possible to achieve maximum throughput of processing one tridiagonal system per clock cycle. The main advantage of the fixed point solvers is the reduced resource usage, which provides the opportunity to maximise coarse grain parallelism by allowing more solver cores to fit onto a device and also increasing the maximum number of pipelined tridiagonal systems each core can solve. As can be expected the amount of memory resources is proportional to the total data width used for the fixed-point designs, whilst the floating-point solver, although 32bits wide, uses significantly more memory resources (BRAM and memory LUTs).

\subsection{Performance}

\begin{table}\label{tab:speed}
\center
\caption{The average time(ms) for computing the solution to tridiagonal systems (N=100) on a desktop CPU and the implemented FPGA Thomas solver.}
\begin{tabular}{lll}
\hline
 & Max Throughput & Min Throughput \\
 \hline
CPU(2.6GHz) & 0.020ms(1x) & 0.020ms(1x) \\
Floating & 0.0012ms(16x) & 0.063ms(0.31x) \\
Fixed[2,30] & 0.00055ms(36x) & 0.040ms(0.50x) \\
Fixed[2,22] & 0.00055ms(36x) & 0.036ms(0.55x) \\
Fixed[2,14] & 0.00057ms(35x) & 0.028ms(0.72x) \\
\hline
\end{tabular}
\end{table}

The latency performance of the solver can be evaluated using Equation \ref{eq:speed} once the implemented FPGA clock speed is known. The floating-point design was only able to achieve a maximum clock frequency of 100MHz whilst the fixed-point designs where able to achieve double this at 200MHz. Therefore although the fixed-point designs may have slightly higher latency in terms of clock cycles, the speed of computation is considerably faster due to the higher clock rate. 

The average time in milliseconds per tridiagonal system is shown in Table \ref{tab:speed} for minimum throughput, a single tridiagonal system, and maximum throughput, where the pipeline is completely full. These results are compared to the average time taken for a 2.6GHz CPU on a top of the range desktop machine. If the solver was to be used for single tridiagonal systems the speed is fractionally less than a top of the range 2.6GHz CPU, but this is without taking advantage of the pipelined design. At maximum throughput it is possible to achieve up to a 36x speed-up and 16x speed-up over a 2.6GHz CPU for FPGA fixed-point and floating-point designs respectively. In terms of cost of computing power the basic \$200 FPGA board used here can outperform, in terms of speed, a \$1000+ desktop computer, as well as also using considerably less power. This is due to the deep pipelining and custom data paths possible on an FPGA.

\section{Implementation for Implicit Finite Difference Schemes}
In this section we evaluate the accuracy of the fixed-point Thomas solvers in the context of options pricing. Pricing options via implicit finite difference methods require one or a system of several tridiagonal equations to be solved. We solve tridiagonal systems arising in the implicit finite difference scheme for European options using the Black-Scholes model.

\subsection{Scaling For Fixed-Point Designs}
As previously discussed the motivation is to use fixed-point arithmetic because it results in smaller and faster designs compared to floating-point. The tridiagonal coefficients for pricing a European option via implicit finite difference are given by:
\begin{align}
\notag
a_n &= -(n^2\sigma^2 - nr)dt \\
\notag
b_n &=  1 + (n^2\sigma^2 + r)dt \\
\notag
c_n &=  -(n^2\sigma^2 + nr)dt\\
\notag
a_N &= Nrdt \\
\notag
b_N &= 1 -(Nr-r)dt\\
\notag
y_n &= \text{payoff}(S_n)
\end{align}

\noindent where $y$ is defined by the initial boundary condition of the problem, in this case the payoff function of the option.

Observing the coefficients $b_n > 1$ $\text{ for all } i \leq N$, two integer bits will be used for the fixed-point representation and $Z = 2$ to ensure no arithmetic overflow. Using Proposition 1 it is possible to show that the coefficients of implicit Black-Scholes pricing the algorithm can be bounded so that $Z = 2$ after applying a basic grid constraint  to bound coefficient values and an appropriate transformation for the $y$ values to meet condition 2 of Proposition \label{prop:tConditions}. This is more formally expressed in Proposition \ref{prop:grid}.

\begin{prop}\label{prop:grid}
Given a Black-Scholes implicit pricing problem, $A_I$, it is possible to ensure that the supremum of the algorithm i.e. all values calculated in the algorithm, $sup(|A_I|) < Z$, given the following grid constraint and suitable linear transform on the problem domain.\\
\begin{equation}\label{eq:gridConstraint}
dt < \frac{1}{\sigma^2N^2 }
\end{equation}
\begin{equation}
\hat{y}_n = f(y_n)
\end{equation}
\begin{equation}\label{eq:transform}
f(y_n) = y_nZ\frac{|b_0| - \|a\|_{\infty}}{(|b_0|+1)\|y\|_{\infty}} 
\end{equation}
\end{prop}

\subsection{Fixed-Point Solver Accuracy}
The fixed-point solver designs are tested over a sample of 5000 randomly selected tridiagonal equations generated by random option pricing problems. The two market dependant parameters, interest rate $r$ and volatility $\sigma$, are randomly chosen for each option sample to generate a new sample of tridiagonal equations to solve, with $r = U[0.01,0.05]$, $\sigma = U[0.10,0.30]$. The finite difference grid parameters are selected to meet the constraint in Equation \eqref{eq:gridConstraint} in the case of maximum market parameter values, which requires $dt = 0.001$. Finally, to meet the final condition for  $\|y\|$ a scaling factor of $\frac{0.45Z}{\|y\|_{\infty}}$ was used, calculated using Equation \eqref{eq:transform}, and to meet this condition the problem was chosen so that $S_{N} = 2$ and $K = 1$.

Table \ref{tab:minRes} gives the expected rounding error for the number of fractional bits used in the fixed-point design. The expected rounding error $e_{rnd}(x,f)$, where $e_{rnd}$ is the rounding error and $f$ is the number of fractional bits, for rounding a floating-point number $x$ to a fixed-point representation is given by:
\begin{equation}\label{eq:exptRndError}
\mathbb{E}(e_{rnd}(x,f)) = \frac{2^{f-1}}{2} \text{ ; } x \in U[0,\infty] .
\end{equation}
\noindent We assumed that the rounding error is uniformly distributed white noise \cite{Barnes1985}. If an error obtained is smaller than this value this indicates that the fixed-point value was rounded to 0, and the actual value is smaller than is possible to represent in the fixed-point representation. Errors within a similar magnitude as the magnitude of the expected error indicate that the fixed-point result is on average as accurate as is possible for the given fixed-point representation.

Figure \ref{fig:fx3230} shows the absolute error with respect to the floating point result for the fixed-point solver using 30 fractional bits. The most striking feature of this plot is how the error resembles the shape of the payoff function indicating that the magnitude of the option price plays a role in the error function. A worst case error function has been derived from the observation that for an option price, $V$ it holds that $V < S_n$, i.e the European option price must be at most the asset price due the the effect of the strike. The maximum error is then a function of the asset price, minimum expected rounding error and $n$ to take into account error prorogation factors through the iterations

\begin{equation}\label{eq:fxerror}
E(S_n) = nS_n\frac{2^{f-1}}{2}.
\end{equation}

\begin{figure}[t]
\begin{center}
\includegraphics[scale=0.3]{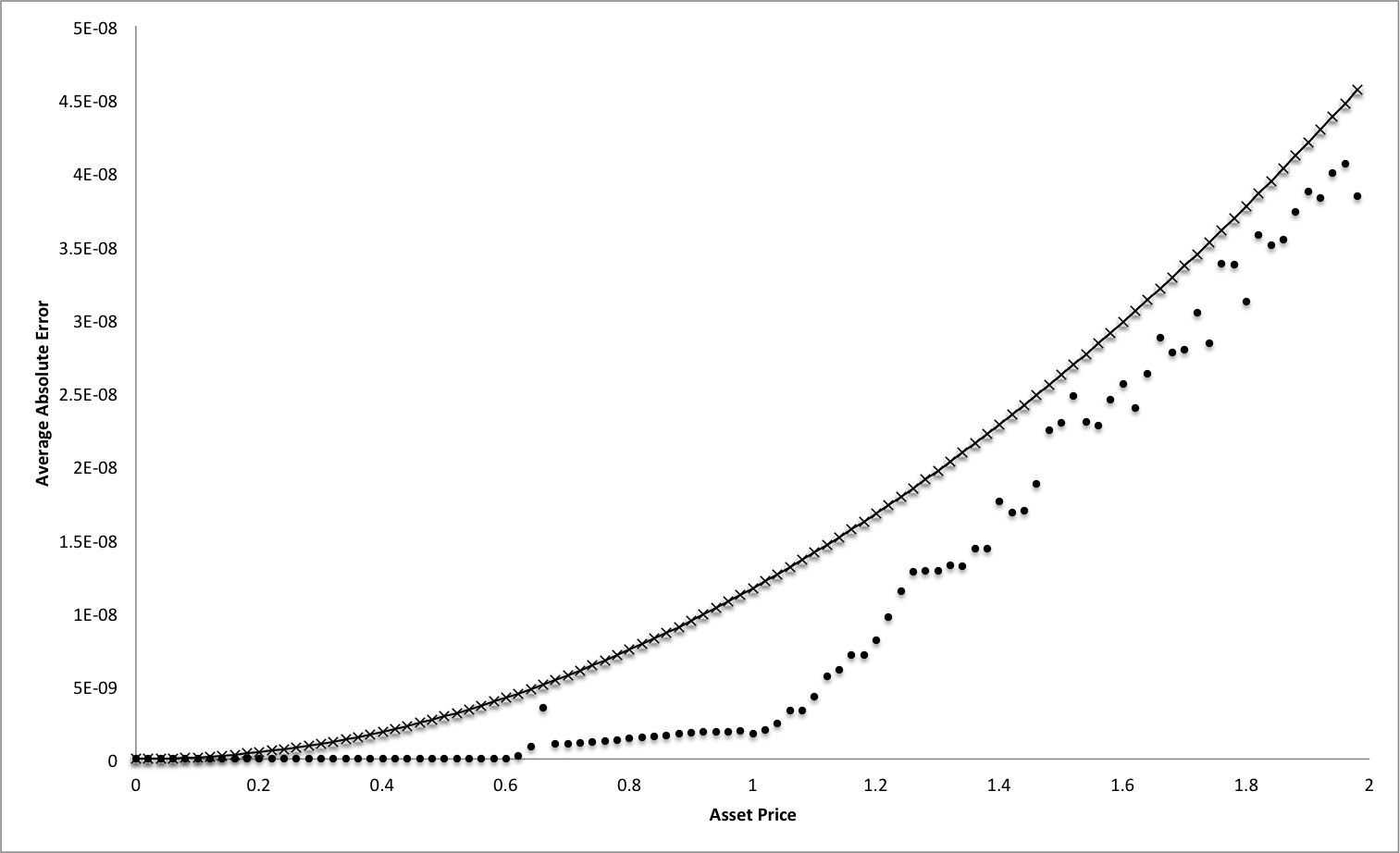}
\caption{Average absolute error over 5000 tridiagonal systems of the fixed-point results using 30 fractional bits with respect to floating-point results ($\bullet$). Estimated maximum error bound using equation \eqref{eq:fxerror}($\times $).}
\end{center}
\label{fig:fx3230}
\end{figure}

\begin{figure}[t]
\begin{center}
\includegraphics[scale=0.3]{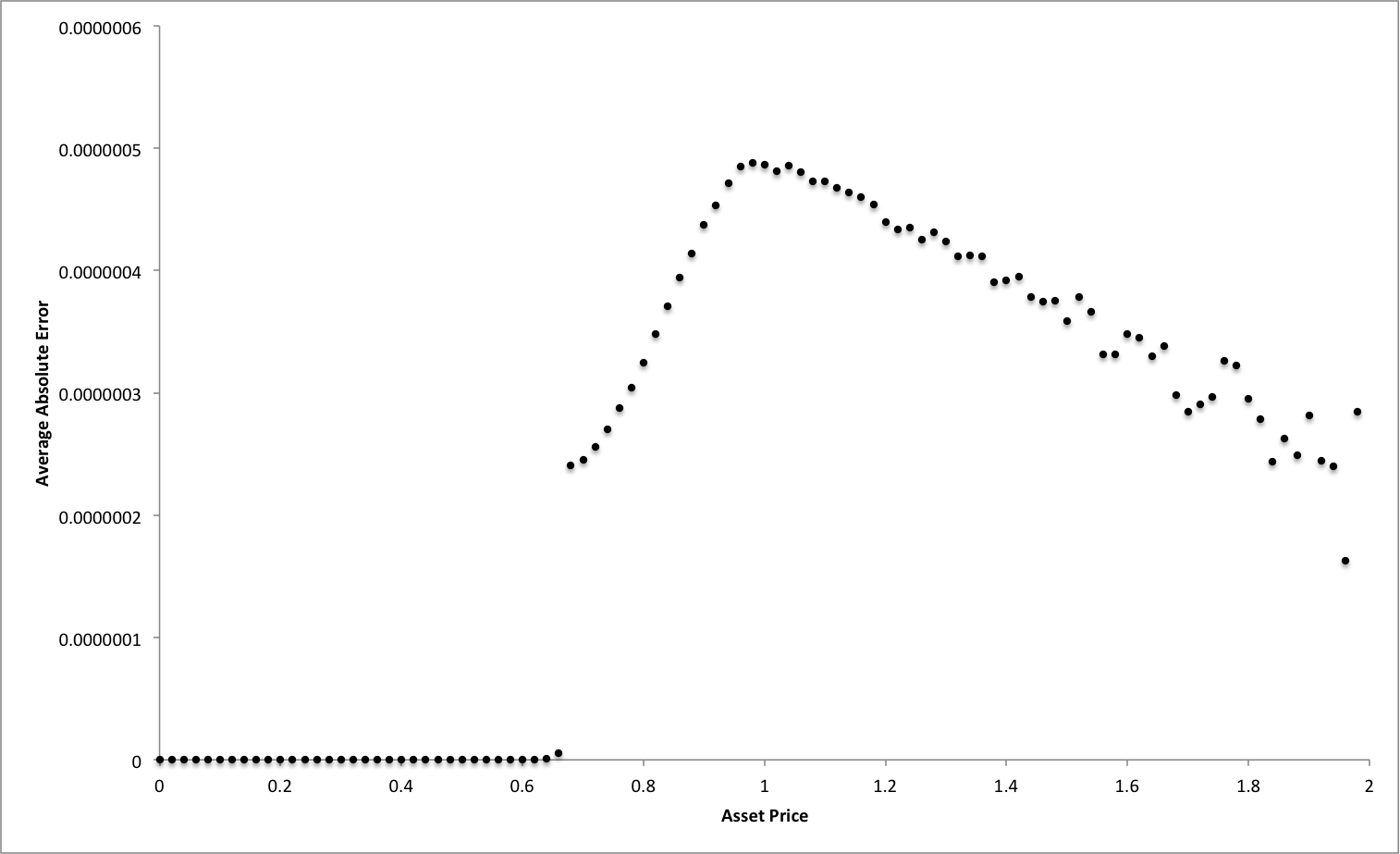}
\caption{Average absolute error over 5000 tridiagonal systems of the fixed-point results using 22 fractional bits with respect to floating-point results ($\bullet$).\label{fig:fx3222}}
\end{center}
\end{figure}

\begin{figure}[t]
\begin{center}
\includegraphics[scale=0.3]{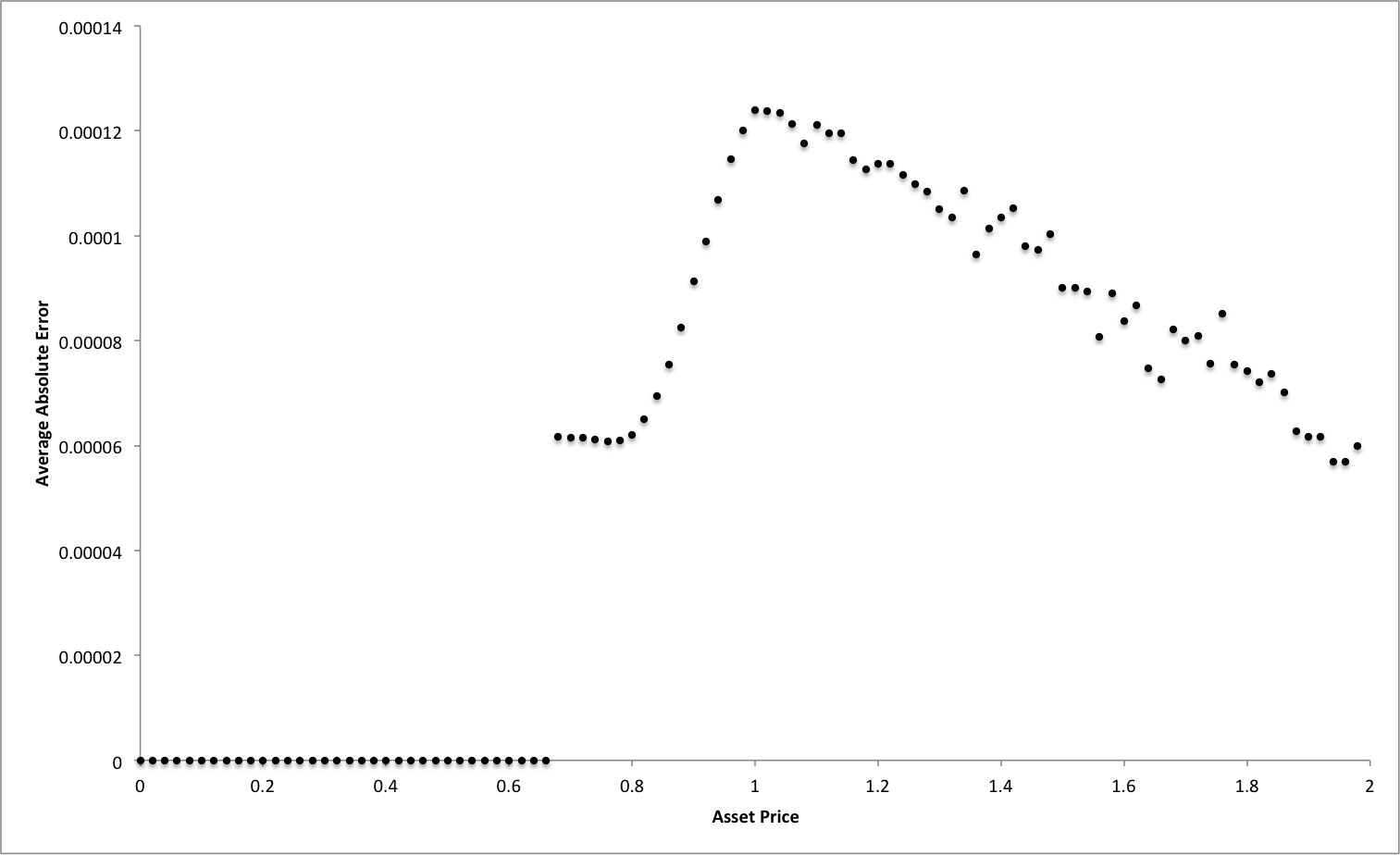}
\caption{Average absolute error over 5000 tridiagonal systems of the fixed-point results using 14 fractional bits with respect to floating-point results.\label{fig:fx3214}}
\end{center}
\end{figure}

\begin{table}\label{tab:minRes}
\center
\caption{Comparison of expected rounding error and maximum absolute error from the FPGA implementation.}
\begin{tabular}{llll}
\hline
 & \multicolumn{3}{c}{Fractional Width} \\
 & 30 & 22 & 14 \\ \hline
Expected Rounding Error & 2.33E-10 & 5.95E-08 & 1.52E-05 \\
Maximum FPGA Error & 4.06E-08 & 4.88E-07 & 1.23E-04 \\
\hline
\end{tabular}
\end{table}

Figures \ref{fig:fx3222} and \ref{fig:fx3214}
 show the absolute errors for the fixed-point solver with 22 and 14 fractional bits respectively. Both of sets of errors show a shape differing from the one observed for 30 fractional bits, with a peak near the strike price followed by a decent. However their respective absolute errors with respect to their minimum fractional resolution are considerably improved, with the largest magnitude of error being of the same order, this is up to 100 times smaller in magnitude than the error predicted by equation \ref{eq:fxerror}. 

These results show that the the fixed-point arithmetic is accurate up to a given decimal place, after which the accuracy begins to degrade. This explains why the 30 fractional bit errors were considerably larger than the respective minimum fractional resolution and in contrast to the 22 and 14 fractional bit implementations.

\section{Conclusions and Future Research}
In this work we proposed and introduced a prototype design for a high performance FPGA based tridiagonal solver. Fixed-point designs can be used to minimise resource usage and obtain higher clock rates compared to floating point designs. When compared to a 2.6GHz CPU on a top of the range desktop it was possible to achieve up to a 36x speed-up and 16x speed-up for the fixed-point and floating-point designs respectively. When solving large sets of independent tridiagonal system the system can be linearly scaled up by adding more solver cores onto the FPGA. For the fixed-point designs the errors introduced in the results due to the limited fractional resolution was investigated. Overall in the implicit option pricing examples the errors were well behaved with the maximum for the 22 and 14 bit fractional representations at 10x that of the expected rounding error, and 50x for 30 fractional bits. In summary, our method can be further integrated into a larger FPGA based implicit pricing system to achieve a high speed and low cost solution for accelerating options pricing.
Future work will investigate improving the accuracy whilst trying to retain high clock rates and low FPGA resource usage by using mixed-precision architectures; for example computing the forward iterations in 16bit fixed-point and the backward iterations in 32bit floating-point. Following up from the theoretical analysis further work with affine arithmetic techniques \cite{Fang2003} it maybe possible to develop a deeper understanding of the fixed-point error behaviour and propagation in the algorithm. Advancing from this work it would also be of interest to investigate the performance of more inherently parallel algorithms such as parallel cyclic reduction on an FPGA.\\

\newpage

\bibliographystyle{spbasic} 
\bibliography{jabref_adi.bib}

\end{document}